\documentclass{sig-alternate}
\usepackage{graphicx}
\usepackage{subfigure}
\usepackage{epstopdf}
\usepackage{color}
\usepackage{amsfonts}
\usepackage{latexsym}

\usepackage{amsthm}
\usepackage{amssymb}
\usepackage{enumerate}
\usepackage{tikz,pgf,pgfplots}
\usepackage{cite}
\usepackage{bm}
\usepackage{dsfont}
\usepackage{xcolor}
\usepackage{amsmath}
\usepackage[vlined,ruled,linesnumbered,onelanguage]{algorithm2e}

\usetikzlibrary{patterns}
\theoremstyle{definition}
\newtheorem{theorem}{\quad Theorem}
\newtheorem{lemma}{\quad Lemma}
\newtheorem{app_lemma}{\quad Lemma}[section]

\newtheorem{assumption}{\quad Assumption}

\theoremstyle{remark}

\newcommand{\superscript}[1]{\ensuremath{^{\textrm{#1}}}}
\def\sharedaffiliation{\end{tabular}\newline\begin{tabular}{c}}

\newfont{\mycrnotice}{ptmr8t at 7pt}
\newfont{\myconfname}{ptmri8t at 7pt}
\let\crnotice\mycrnotice%
\let\confname\myconfname%

\permission{Permission to make digital or hard copies of all or part of this work for personal or classroom use is granted without fee provided that copies are not made or distributed for profit or commercial advantage and that copies bear this notice and the full citation on the first page. Copyrights for components of this work owned by others than ACM must be honored. Abstracting with credit is permitted. To copy otherwise, or republish, to post on servers or to redistribute to lists, requires prior specific permission and/or a fee. Request permissions from permissions@acm.org.}
\conferenceinfo{MobiHoc'15,}{June 22--25, 2015, Hangzhou, China.}
\copyrightetc{Copyright \copyright~2015 ACM \the\acmcopyr}
\crdata{978-1-4503-3489-1/15/06\ ...\$15.00.\\
http://dx.doi.org/10.1145/2746285.2746305}

\begin{document}
\title{A High Reliability Asymptotic Approach for Packet Inter-Delivery Time Optimization in Cyber-Physical Systems}
\numberofauthors{1}
\author{ 
\alignauthor Xueying Guo\superscript{*},
Rahul Singh\superscript{\dag},
P. R. Kumar\superscript{\dag},
Zhisheng Niu\superscript{*}
\sharedaffiliation
\affaddr{\superscript{*}Tsinghua National Laboratory for Information Science and Technology, Tsinghua University, P. R. China}\\
\affaddr{\superscript{\dag}Department of Electrical and Computer Engineering, Texas A\&M University, USA}\\
\email{guo-xy11@mails.tsinghua.edu.cn, \{rsingh1,prk\}@tamu.edu, niuzhs@tsinghua.edu.cn}
}
\maketitle

\begin{abstract}
In cyber-physical systems such as automobiles, 
measurement data from sensor nodes should be delivered to other consumer nodes such as actuators in a regular fashion. 
But, in practical systems over  unreliable media such as wireless, 
it is a significant challenge to guarantee small enough inter-delivery times for different clients with heterogeneous channel conditions and inter-delivery requirements. 
In this paper, we design scheduling policies aiming at satisfying the inter-delivery requirements of such clients. 
We formulate the problem as a risk-sensitive Markov Decision Process (MDP). Although the resulting  problem involves an infinite state space, we first prove that there is an equivalent MDP involving only a finite number of states. Then we prove the existence of a stationary optimal policy and establish an algorithm to compute it in a finite number of steps.

However, the bane of this and many similar problems is the resulting complexity, and, in an attempt to make fundamental progress, we further propose a new \textit{high reliability asymptotic} approach.  
In essence, this approach considers the scenario when the channel failure probabilities for different clients are of the same order, and asymptotically approach zero. We thus proceed to determine the asymptotically optimal policy: 
in a two-client scenario, we show that the asymptotically optimal policy is a ``modified least time-to-go" policy, which is intuitively appealing and easily implementable; in the general multi-client scenario, we are led to an SN policy, and we develop an algorithm of low computational complexity to obtain it. Simulation results show that the resulting policies perform well even in the pre-asymptotic regime with moderate failure probabilities.
\end{abstract}

\category{C.2.1}{Network Architecture and Design}{Wireless Communication}

\keywords{Wireless Sensor Networks; Scheduling; Packet Inter-Delivery Time; High Reliability Asymptotic Approach}

\section{Introduction}
Delay and throughput have long been regarded as important quality of service (QoS) metrics \cite{Xiong2011,Neely2005,Guo2013}. However, with the increasing deployment of real-time applications such as sensor networks, and surveillance applications over unreliable media such as wireless, guaranteeing small enough inter-delivery times between  packets becomes important \cite{atilla1,atilla2,Rahul2015,r2,r3}. As an example, consider an in-vehicle wireless sensor network illustrated in Fig. \ref{fg:car}.
In-vehicle wireless sensor networks have been drawing increasing attention recently since they can significantly reduce the costs, reduce weight of the wiring harness and hence increase fuel efficiency, and are extensible and scalable, as compared to the wired in-vehicle networks\cite{Elbatt2006,Tsai2007}.
Such a cyber-physical system features several (about a hundred) wireless sensor nodes monitoring processes such as temperature and pressure, and continually transmitting their measurements to controllers which then choose appropriate actuation signals. 
In these systems, one is allowed to control the arrival process since outdated packets containing old sensor measurements can be
replaced by newer packets.
Since a large gap between updates can lead to system instability, inter-delivery times of these packets is an important QoS metric. Different clients may have different channel
conditions and inter-delivery requirements, which further complicates the problem.

\begin{figure}[!t]
	\centering
	\includegraphics[width=0.3\textwidth]{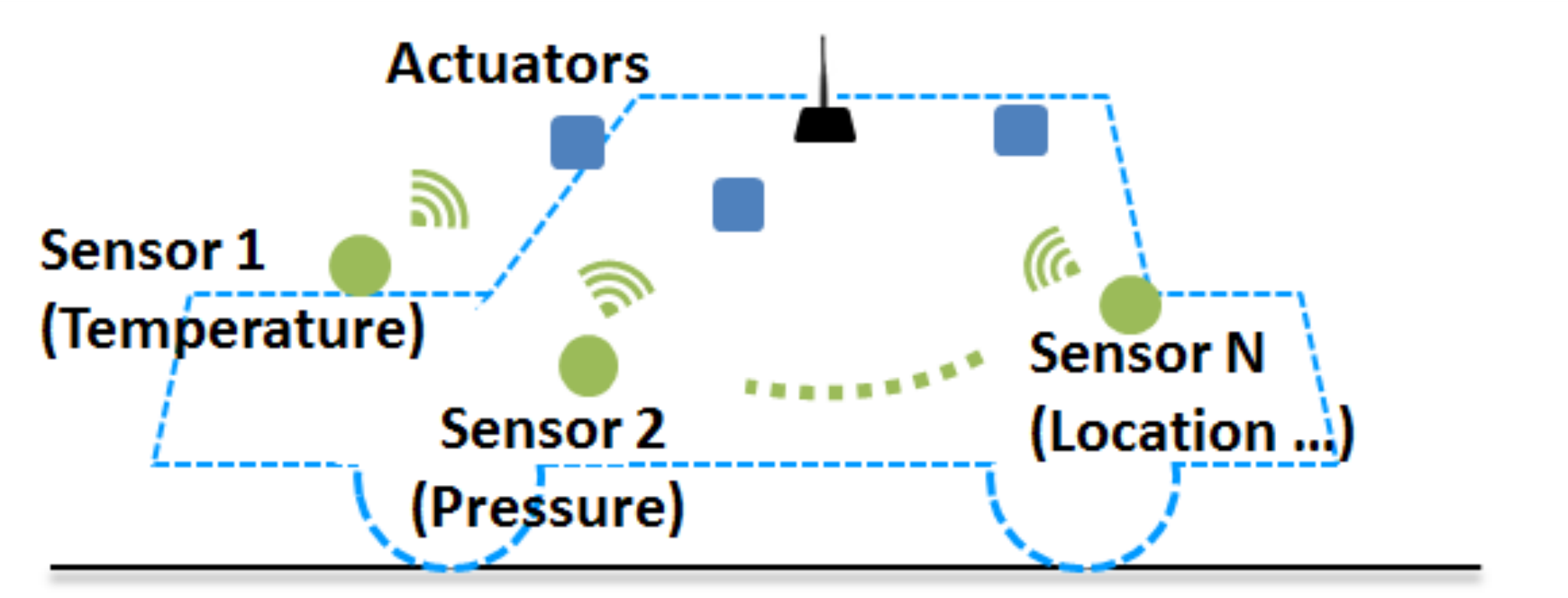}
	\caption{An in-vehicular network  with an access point and several wirelessly connected sensors and actuators.}
	\label{fg:car}
\end{figure}

In this paper, our goal is to design scheduling policies that decide which client's packet to transmit in each time slot in such systems, so as to guarantee small enough inter-delivery times between packets. 

To penalize severely the deviations in inter-delivery times that are larger than a certain threshold, we consider the exponential cost function, 
\begin{align*}
\mathrm{E}\left[\mathrm{exp}\left(\theta \cdot \sum_{n=1}^{N} (D^{(n)}-\tau_n)^+\right)\right], 
\end{align*}
where $\theta>0$ is a risk-aversion parameter. Here, $D^{(n)}$ is the inter-delivery time of client $n$, $\tau_n$ is a specified inter-delivery threshold for client $n$, and $(a)^+:=\max\{0,a\}$.

We formulate the optimization problem as a risk-sensitive Markov Decision Process (MDP) \cite{Howard1972,chung1987,Fleming1997,Marcus1997,Cadena1999,Jaquette1976}. Though it is over an infinite state space, we show that there is an equivalent MDP that involves only a finite state space. 
For this equivalent MDP, we then prove the existence of a stationary optimal policy, and obtain an algorithm which determines the optimal policy in a finite number of steps. 

The significant challenge of this and other modeling efforts is however the complexity of determining an optimal solution to the MDP, which is excessively large even for systems with a moderate number
of states (e.g., a hundred nodes), as is of interest in many applications. 
To address this critical challenge, we further propose a new approach to channel modeling which we call the ``high reliability asymptotic regime". 
In essence, this approach considers the scenario when the channel failure probabilities for different clients are of the same order, and asymptotically approach zero. We then proceed to determine the asymptotically optimal policy. 
Such a policy is expected to provide near optimal performance even when the channel failure probabilities are non-zero and range from small to moderate values. Philosophically, this approach can be regarded as similar to studying the high SNR asymptotics in network information theory \cite{avestimehr2011wireless}. 

In the case where there are two clients, the asymptotically policy has a very appealing structure, lending support to this approach.
The asymptotically optimal policy is a structurally clean   ``modified-least-time-to-go" policy, which is both intuitively appealing and easily implementable. 

Our interest, motivated by cyber-physical systems applications, is however in large systems with several sensors and actuators.
For this general multi-client scenario, 
the asymptotic approach leads to an SN policy, and an algorithm of relatively low computational complexity to obtain it. 
The success of such an asymptotic approach however depends on its performance in the pre-asymptotic regime. We present simulation results showing that this is indeed the case.

The rest of the paper is organized as follows. We review the related works in Section 2. We present the system model in Section 3. We formulate the problem as a risk-sensitive MDP in Section 4, and reduce it to a finite state problem in Section 5. 
We prove the existence of a stationary optimal policy, and apply the classic risk-sensitive MDP approach in Section 6. We propose the high reliability asymptotic approach in Section 7. In Section 8, we design asymptotically optimal policies by analyzing the high reliability asymptote. Simulation results are presented in Section 9, followed by conclusions in Section 10.

\section{Related Work}
Li, Eryilmaz and Li \cite{atilla1} and Li, Li and Eryilmaz \cite{atilla2} are apparently the first to consider inter-delivery time as a performance metric. 
These works analyzed this metric in the context of queueing systems, where the relevant trade-off is between stabilizing the queues and minimizing the sum of the inter-delivery times over all the clients. 
However, the situation is very different in wireless sensor networks where packets contain sensor measurements and one can simply replace older packets by newer packets, and thus resulting in no queues. 
Sadi and Ergen \cite{Sadi2013} have pointed out the periodic nature of sensor nodes in intra-vehicular wireless sensor networks.
In another relevant work, Singh, Guo, and Kumar \cite{Rahul2015} have addressed the issue of trading off higher throughput for better performance with respect to variations in the inter-delivery times. 
However, this work does not allow tunable and heterogeneous inter-delivery requirements, and a buffer is maintained so as to mitigate the influence of variations in inter-delivery times. 
Guo, Singh, Kumar, and Niu \cite{Guo2015_ICC} have further combined the inter-delivery time requirement with the system energy-efficiency.
Singh and Stolyar \cite{r1} have shown that the service process under the Max Weight scheduling is asymptotically smooth. 



The study of risk-sensitive MDPs dates back to Howard and Matheson  \cite{Howard1972}. 
There has been considerable work on proving the existence of stationary optimal policy in different conditions \cite{Fleming1997,Marcus1997,Cadena1999}. 
However, Chung and Sobel \cite{chung1987} have pointed out that, in contrast to the risk-neutral MDPs, even with a discounted cost, a stationary optimal policy need not exist for a general risk-sensitive MDP. 
This introduces an additional challenge to our study. 
In wireless communications, Altman et al \cite{Altman2011} have applied risk-sensitive MDP techniques to design power control strategies aiming at minimizing the delivery failure probability in delay tolerant networks.

Avestimehr et al \cite{avestimehr2011wireless} have proposed a deterministic channel model to study the high SNR asymptotics in the field of network information theory, thus obtaining constant gap approximations to the capacity of wireless networks. This has led to near-optimal and easy-to-implement communication schemes for Gaussian relay networks. 
Kittipiyakul et al \cite{Kittipiyakul2009} and Zhang et al \cite{Zhang2012} have also employed such an asymptotic approach to analyze error performance in fading channels. 


\section{System Model}\label{model}

Consider a system with $N$ wireless sensors and one access point (AP). Time is discretized  into slots.
The AP broadcasts a control message at the beginning of each time-slot to announce which sensor can transmit in the slot. The assigned sensor then transmits a packet. The size of a time slot is the time required for the AP to send the control message plus the time for a client to prepare and transmit a packet.
It is assumed that the wireless channel connecting sensor $n$ and the AP has a \emph{channel reliability} of $p_n\in(0,1)$, which can be taken to be the probability that the control message from the AP and the transmission from client $n$ are both successful. The system model can be generalized to take into account more general fading models.

The \textit{QoS} requirement for client $n$ is modeled through a specified value for the inter-delivery threshold $\tau_n$.
The cost incurred by the system in $T$ time-steps is modeled as,
\begin{align}\label{eq:cost define origninal}
\mathrm{E}\!\bigg[\!\mathrm{exp}\!\bigg(\!\theta\! \sum_{n=1}^N \!\! \sum_{i=1}^{M_T^{(n)}}\!(\!D_i^{(n\!)}\!\!-\!\tau_n)^+\!\! +\!(T\!-\!t_{\!D_{\!M_T^{(n)}}^{\!(n\!)}}\!\!\!-\tau_n)^+\!\! \bigg)\!\bigg]\!,
\end{align}
where $D_i^{(n)}$ is the time between the $(i\!-\!1)$-th and $i$-th packet deliveries of client $n$, $M_T^{(n)}$ is the number of packets delivered for the $n$-th client by time $T$, 
$t_{D_i^{(n)}}$ is the time slot in which the $i$-th packet for client $n$ is delivered, and $(a)^+:=\max\{a,0\}$. 
The last term is included since, otherwise, the policy of never making any transmission at all will result in the least cost.
The parameter $\theta>0$ leads to a risk-aversion problem. 
The goal of the AP is to decide the client to transmit in each time slot, in order to minimize the above cost.

\section{Problem Formulation}\label{se:problem formulation}
We first describe the notations used: Vectors will be in bold font, e.g., $\boldsymbol{\tau}: =\left(\tau_1,\ldots,\tau_N\right)$ and $\mathbf{x}:=(x_1,\cdots,x_N)$. Define $a_n\wedge b_n:=\min\{a_n,b_n\}$, and $\mathbf{a}\wedge \mathbf{b}: = \left(a_1 \wedge b_1,\ldots, a_N\wedge b_N\right)$.

The system \textit{state} at time $t$ is denoted as
$$X(t):= \left(X_1(t),\ldots,X_N(t)\right),$$ where $X_n(t)$ is the time 
elapsed since the 
most recent packet delivered by client $n$. Thus, the \textit{state space} is $\{0,1,\cdots\}^N$, which is finite for the finite time horizon problem, but exponentially growing to infinity as the horizon increases. 
Let \textit{control} $U(t)$ denote the client transmitting
in time slot $t$.
The system state evolves as, 
\begin{align*}
X_n(t+1) =
\begin{cases}
0 \mbox{ ~~if a packet is delivered for client } n \mbox{ in slot } t,\\
X_n(t) + 1    \mbox{\quad \quad otherwise }.
\end{cases}
\end{align*}


As a consequence, the system forms a controlled Markov Chain, with transition probabilities, 
\begin{align*}
\nonumber
&\mathrm{P}\left[X(t+1)=\mathbf{y}\big|X(t)=\mathbf{x},U(t)=u\right] \\
=&
\begin{cases}
p_u & \mbox{if } \mathbf{y}=(x_1\!+\!1,\!\cdots\!,x_{u\!-\!1}\!+\!1,0,x_{u\!+\!1}\!+\!1,\!\cdots\!,x_{\!N}\!+\!1), \\
1-p_u & \mbox{if } \mathbf{y}=\mathbf{x}+\mathbf{1}, \\
0 & \mbox{otherwise}, 
\end{cases}
\end{align*}
where  $\mathbf{1}:=(1,\cdots,1)$. 
The $T$-horizon optimal cost-to-go from initial state $\mathbf{x}$ is given by,
\begin{align}\label{eq:cost define MDP1 v_T}
\nonumber
V_T (\mathbf{x})&:=\min _{\pi} \mathrm{E}_\pi \bigg[\exp\bigg(\theta \sum_{t=0}^{T-1} \sum_{n=1}^{N}  \left(X_n(t) + 1-\tau_n \right)^{+} \\
&\cdot
\mathds{1}\{X_n(t+1) = 0\}\bigg)\Big|X(0)=\mathbf{x}\bigg],
\end{align}
where $\mathds{1}\{\cdot\}$ is the indicator function, and assuming  $X(T):=\mathbf{0}$ so as to recover the last term in the cost \eqref{eq:cost define origninal}. 
Here, the minimization is taken over all history-dependent scheduling policies $\pi$. 
Our goal is to design an optimal history-dependent policy that achieves the optimal cost-to-go $V_T(\mathbf{x})$ for any initial state $\mathbf{x}$. 

Later in Section \ref{se:risk-sensitive MDP}, equation \eqref{eq:long term cost define}, we will consider the infinite horizon cost $J(\pi,\mathbf{x})$ when $T\rightarrow\infty$.  
\section{Reduction to Finite State Problem}
We denote the problem in Section \ref{se:problem formulation} as \textit{MDP-1}. 
We now show that it is equivalent to another finite-state problem.  

It directly follows from \eqref{eq:cost define MDP1 v_T} that the DP recursive relationship of the optimal cost-to-go functions in MDP-1 is: 
\begin{align}\label{eq:Recursive Relationship v1}
\nonumber
V_T(\mathbf{x})&=\min_{n} \big\{
p_n \exp\left({\theta \left(x_n+1-\tau_n\right)^+}\right) V_{T-1}\big(\mathcal{S}_n(\mathbf{x})\big) \\
& + (1\!-\!p_n) V_{T-1}\left(\mathbf{x}+\mathbf{1}\right)
\big\}, 
\end{align}
where
\begin{align}\label{eq:Snx}
	\mathcal{S}_n(\mathbf{x}):=(x_1\!+\!1,\!\cdots\!,x_{n\!-\!1}\!+\!1,0,x_{n\!+\!1}\!+\!1,\!\cdots\!,x_{\!N}\!+\!1),
\end{align}
i.e., the state that succeeds the state $\mathbf{x}$ in the event of a successful transmission for client $n$. 

\begin{lemma}\label{th:MDP-1 Lemma}
For MDP-1, the following results hold:  
\begin{enumerate}[1)]
\item For all $n\in\{1,\cdots,N\}$, and  $\forall x_1,\cdots,x_N\geq 0$, 
\begin{align}\label{eq:fundamental claim}
\nonumber
& V_T\big(x_1,\cdots,x_n+\tau_n,\cdots,x_N\big) \\ 
&= \exp\left({\theta x_n}\right) \cdot V_T\big(x_1,\cdots,\tau_n,\cdots,x_N\big).
\end{align}
Further, the optimal controls in the two states, \\ $\left(x_1,\cdots,x_n\!+\!\tau_n,\cdots,x_N\right)$ and  $\left(x_1,\cdots,\tau_n,\cdots,x_N\right)$, are the same.
\item The optimal cost function starting with any system state $\mathbf{x}$ such that $x_n\leq \tau_n, \forall n$ satisfies:
\begin{align}\label{eq:equivalent recursive relationship}
\nonumber
V_T(\mathbf{x}\!)&=\!\exp\!\Big(\!\theta \! \sum_{n=1}^{N}\! \mathds{1}\!\{x_n=\tau_n\} \!\!\Big)
\min_{u} \!\Big\{
p_u  V_{T\!-\!1}\left(\mathcal{S}_u(\mathbf{x})\!\wedge \!\boldsymbol{\tau}\right) \\
&+ (1-p_u) V_{T-1}\left(\left(\mathbf{x}+\mathbf{1}\right)\wedge \boldsymbol{\tau}\right)
\Big\}, 
\end{align}
where $\mathcal{S}_u(\mathbf{x})$ is as in \eqref{eq:Snx};
\item  
$Y(t):= X(t) \wedge \boldsymbol{\tau}$ is a Markov Decision Process, i.e., 
\begin{align*}
&\mathrm{P}\big[Y(t+1)\big|Y(t),\cdots,Y(0), U(t),\cdots,U(0) \big]\\ =&\mathrm{P}\big[Y(t+1)\big|Y(t), U(t)\big].
\end{align*}
\end{enumerate}	
\end{lemma}
\begin{proof}
The proof is omitted due to space constraints.
\end{proof}

Now, we construct a new \textit{MDP-2}, and show in Theorem \ref{th:equivalent theorem} that it is equivalent to MDP-$1$ in an appropriate sense. By slightly abusing notation, we still use the symbols $Y(t)$ and $U(t)$ for state and client-to-transmit. 

Let us associate a state variable $Y_n(t)$ with each client $n$, with $Y_n(0)\in \{0,1,\cdots, \tau_n \}$, which evolves as, 
\begin{align*}
Y_n(t+1) =
\begin{cases}
0 \mbox{~~if a packet delivered for client } n \mbox{ in slot } t,\\
\left(Y_n(t) + 1 \right) \wedge \tau_n \mbox{ \quad \quad  otherwise }.
\end{cases}
\end{align*}
{Then the system \textit{state space} is $\mathbb{Y}:=\prod_{n=1}^{N}\left\{0,1,\cdots,\tau_n \right\}$}, which is finite, even for the infinite time horizon problem. 
The ~transition ~probabilities ~of ~the ~process \\ $Y(t): =\left(Y_1(t),\cdots,Y_N(t)\right) $ depend on the control $U(t)$ as, 
\begin{align} \label{eq:system evolve}
&\mathrm{P}\big[Y(t+1)=\mathbf{y}\big|Y(t)=\mathbf{x},U(t)=u\big] \\ \nonumber 
= & \left\{
\begin{array}{ll}
p_u & \mbox{if } \mathbf{y}=\mathcal{S}_u(\mathbf{x})\wedge\boldsymbol{\tau}, \\
1-p_u & \mbox{if } \mathbf{y}=\left(\mathbf{x}+\mathbf{1}\right)\wedge\boldsymbol{\tau} , \\
0 & \mbox{otherwise}, 
\end{array}
\right.
\end{align}
where $\mathcal{S}_u(\mathbf{x})$ is as in \eqref{eq:Snx}.

We associate the following cost to the system with starting state $\mathbf{x}\in \mathbb{Y}$, when policy $\pi$ is applied: 
\begin{align} \label{eq:cost define MDP2 V_T^Pi}
V_T^{\pi}(\mathbf{x})=
\mathrm{E}_{\pi} \!\Big[\exp \Big(\!\theta \! \sum_{t=0}^{T-1}\!\sum_{n=1}^{N} \mathds{1}\!\!\left\{Y_n(t)=\tau_n\right\}\!\!
\Big)\!\Big|Y\!(0)=\mathbf{x} \Big] .
\end{align}
The optimal cost-to-go function is, 
\begin{align}
\tilde{V}_T (\mathbf{x}):= \min_{\pi} V_T^\pi(\mathbf{x}), \forall \mathbf{x}\in \mathbb{Y}. 
\end{align}
Here, the superscript tilde is to distinguish it from the optimal cost for the MDP-$1$ in \eqref{eq:cost define MDP1 v_T}.

\begin{theorem}\label{th:equivalent theorem}
The MDP-2 is equivalent to MDP-1 in the following senses:
\begin{enumerate}[1)]
\item The optimal cost-to-go functions of the two MDPs are equal in each time slot $t$ for any starting state  $\mathbf{x}$ such that $x_n\leq \tau_n,\forall n$, i.e., 
\begin{align*}
V_T(\mathbf{x})=\tilde{V}_T(\mathbf{x}), \forall \mathbf{x} \in \mathbb{Y};
\end{align*}
\item Any optimal control for MDP-$1$ in state  $\mathbf{x}$ is also optimal for MDP-$2$ in state  $\mathbf{x}\wedge\boldsymbol{\tau}$, and conversely.
\end{enumerate}

\end{theorem}
\begin{proof}
The DP recursion for the optimal cost in MDP-$2$ is
\begin{align}\label{eq:recursive MDP2}
\tilde{V}_{\!T}\!(\mathbf{x}\!)\!=\! \exp\!\Big(\!\theta \! \sum_{n=1}^{N}\!\mathds{1}\!\{x_n\!=\!\tau_n\!\} \! \!\Big) \! \min_{u}\!\Big\{\! \! \sum_{\mathbf{y}}\!P_{\!u}(\!\mathbf{x},\mathbf{y}\!) \tilde{V}_{T\!-\!1}(\mathbf{y}\!)\!
\Big\}\!, 
\end{align}
where $P_u(\mathbf{x},\mathbf{y}) := \mathrm{P}\big[Y(t+1)=\mathbf{y}\big|Y(t)=\mathbf{x},U(t)=u\big] $. 
Recalling \eqref{eq:system evolve}, we note that 
the r.h.s. of \eqref{eq:equivalent recursive relationship} and the r.h.s. of \eqref{eq:recursive MDP2} evolve in exactly the same way. 
Thus, the optimal cost for MDP-$1$, i.e., $V_T(\mathbf{x})$, and the optimal cost for MDP-$2$, i.e., $\tilde{V}_T(\mathbf{x})$, have identical recursive relationships for $\mathbf{x}\in \mathbb{Y}$. Consequently, statement 1) follows. 

In addition, due to the identical recursive relationships, the optimal controls at any state $\mathbf{x}\in \mathbb{Y}$ for the two systems are also identical. Combining this with the first statement of Lemma \ref{th:MDP-1 Lemma}, we obtain statement 2). 
\end{proof}

\section{The Risk-sensitive Approach} \label{se:risk-sensitive MDP}
The great advantage of the equivalent MDP-$2$ is that its state space is finite. 
Following Theorem \ref{th:equivalent theorem}, we focus exclusively on MDP-$2$ in the following. 

To consider long-term operation, we define
the\emph{ (risk-sensitive infinite horizon) average cost} under policy $\pi$ starting at state $\mathbf{x}$,
\begin{align}\label{eq:long term cost define}
J(\pi,\mathbf{x}):=\limsup_{T\rightarrow\infty} \frac{1}{\theta}\cdot \frac{1}{T}
\ln V_T^{\pi}(\mathbf{x}), \forall\mathbf{x}\in\mathbb{Y},
\end{align}
where $V_T^\pi(\mathbf{x})$ is as in \eqref{eq:cost define MDP2 V_T^Pi}. 

A \textit{stationary policy} is one that decides the current control (which client to transmit) by the current system state. Thus, it can be described 
by a map $f$ from state space $\mathbb{Y}$ to control set $\{1,\cdots,N\}$, i.e., the control $U(t)=f\left(Y(t)\right)$. 

We now define the class of \textit{Non-Exclusionary (NE)} policies as those stationary policies which do not serve a client $n$ when the system state $\mathbf{x}$ is $(\tau_1,\cdots,\tau_{n-1},0,\tau_{n+1},\cdots,\tau_N)$.
It is shown in Appendix \ref{se:policies of interest} that for any non-NE policy, either there is an NE policy out-performing it, or it is trivial to obtain the cost associated with it. 
Thus, we focus on NE policies.

In the following, we use the standard notations of transient/ non-transient states and communicating classes \cite{Puterman1994}. 
For a stationary policy $f$, let $\boldsymbol{\mathrm{P}}^f$ denote its \textit{transition probability matrix}, 
\begin{align*}
(\boldsymbol{\mathrm{P}}^f)_{\mathbf{x},\mathbf{y}}:=\mathrm{P}\left[Y(t+1)=\mathbf{y}\big|Y(t)=\mathbf{x},U(t)=f(\mathbf{x})\right].
\end{align*}
Further, let $\boldsymbol{\mathrm{L}}^f$ denote the \textit{dis-utility matrix} of $f$, such that, 
\begin{align*}
(\boldsymbol{\mathrm{L}}^f)_{\mathbf{x},\mathbf{y}}:=\exp\! \Big( \theta   \sum_{n=1}^{N}\!\mathds{1}\!\{x_n = \tau_n \} \! \Big) \cdot  (\boldsymbol{\mathrm{P}}^f)_{\mathbf{x},\mathbf{y}},\forall \mathbf{x}, \mathbf{y}\in \mathbb{Y}.
\end{align*}
Let $\rho(\boldsymbol{\mathrm{L}}^f)$ be its spectral radius. 

\begin{lemma}\label{th:doeblin condition lemma}
Consider any NE policy $f$,
\begin{enumerate}[1)]
\item There is exactly one non-transient communicating class, which includes the state  $\left(\tau_1,\ldots,\tau_N\right) $.
\item For any transient state $\mathbf{y}$, we have, $(\boldsymbol{\mathrm{P}}^f)_{\mathbf{y},\mathbf{y}}=0$.
\item Assuming that there is only one communicating class (and thus non-transient) when $f$ is applied, we have,   
\begin{align}\label{eq:value for stationary policy}
J(f,\mathbf{y})=\rho(\boldsymbol{\mathrm{L}}^f), ~\forall \mathbf{y}.
\end{align}
\end{enumerate}
\end{lemma}

\begin{proof}
Since $p_n<1$ for each client $n$, there is a positive probability that there will be no packet deliveries for $\tau_{\max}:=\max_{n=1}^N \tau_n$  time slots, and so system state $(\tau_1,\cdots,\tau_N)$ is reachable from any state.

We now prove 2). From \eqref{eq:system evolve}, we have,
$(\boldsymbol{\mathrm{P}}^f)_{\mathbf{y},\mathbf{y}}\neq 0$  implies either $\mathbf{y}=\left(\tau_1,\cdots,\tau_N\right)$, or $\exists n$ and $\mathbf{y}$ such that $f(\mathbf{y})=n$ and $y_{n}=0, y_{l}=\tau_l, \forall l\neq n.$
However, the state $\left(\tau_1,\ldots,\tau_N\right) $ in the former case is non-transient, while the latter condition is ruled out because of the property of the NE policy.

Statement 3) is proved by matrix analysis techniques, and the proof is omitted here due to space constraints. 
\end{proof}

In the following, we assume that for any NE policy, there is only one self-communicating class. This assumption is not restrictive considering the first statement in Lemma \ref{th:doeblin condition lemma}, since we at least can simply restrict the state space to the one non-transient communicating class. (In addition, by the second statement of Lemma \ref{th:doeblin condition lemma}, it follows that there is no one-state transient communicating class.) 
As a result, the average cost of any NE policy can be obtained by \eqref{eq:value for stationary policy}. 

Next, let $p_{\max}:=\max_{n=1}^N p_n$ and $\tau_{\max}:=\max_{n=1}^N \tau_n$. We further denote  $K:=\big\lceil \tau_{\max}(1-p_{\max})^{-\tau_{\max}}\big\rceil.$
\begin{theorem}\label{th:optimality equation} 
Let 
$$\theta_\text{th}:=\frac{\ln(K+1)-\ln(K)}{2N\left(K+1\right)}.$$ 
For the infinite-horizon MDP-$2$, 
\begin{enumerate}[1)]
\item There exists a stationary optimal policy when $\theta<\theta_\text{th}$; 
\item Further, this stationary optimal policy can be computed in a finite number of steps. 
\end{enumerate}
\end{theorem}
\begin{proof}
Denote by $T_\text{Db}(\mathbf{x})$ the first passage time from state $\mathbf{x}$ to $\left(\tau_1,\cdots,\tau_N\right)$, i.e.,
\begin{align*}
T_\text{Db}(\mathbf{x}):=\min\!\big\{t>0\big|Y(t)=\left(\tau_1,\cdots,\tau_N\right),Y(0)=\mathbf{x}\big\}.
\end{align*}
We first prove the simultaneous Doeblin condition \cite{Cadena1999}, i.e., that for any stationary policy $f$, 
\begin{align}\label{eq:simutaneous doblin condition _E[T]}
\mathrm{E}_f[T_{\text{Db}}(\mathbf{x})
]\leq K,~\forall \mathbf{x}. 
\end{align}
Note that for any initial state $\mathbf{x}$, the state $\boldsymbol{\tau}$ will be hit if there are $\tau_{max}$ successive transmission failures. The probability of this event is $\geq (1-p_{\max})^{\tau_{\max}}$. Consider the probability that the state $\tau$ is hit within $j \tau_\text{max}$ time slots, then,
\begin{align*}
\mathrm{E}_f \left[T_\text{Db}(\mathbf{x})\right]&\leq
\sum_{j=1}^{+\infty}  j \tau_{\max}
\frac{(1\!-\!p_{\max})^{\tau_{\max}}}{\left[1\!-\!\left(1\!-\!p_{\max}\right)^{\tau_{\max}}\right]^{-(j-1)}}  \\
&=\frac{\tau_{\max}}{(1-p_{\max})^{\tau_{\max}}}.
\end{align*}
This proves \eqref{eq:simutaneous doblin condition _E[T]}. 
Consequently, statement 1) is proved by combining \eqref{eq:simutaneous doblin condition _E[T]} with the Theorem 3.1 in \cite{Cadena1999}. 

In addition, since we can obtain the average cost of any stationary policy (the cost of NE policy by Lemma \ref{th:doeblin condition lemma} and the cost of a non-NE policy by Appendix \ref{se:policies of interest}), and since there are a finite number of possible stationary policies (resulting from finite state space and control set), statement 2) follows. 
\end{proof}

\section{\!\!\!The High Reliability Asymptotic Approach}
 The risk-sensitive approach faces significant challenges on the issue of computational complexity: For MDP-2, denote the cardinality of the state space $\mathbb{Y}$ by $|\mathbb{Y}|$. Then, there are in total $|\mathbb{Y}|^N$ policies, each of which requires calculating the spectral radius of a $[0,1]^{|\mathbb{Y}|}\times[0,1]^{|\mathbb{Y}|}$ matrix to obtain the corresponding average cost (see Lemma \ref{th:doeblin condition lemma}). Further, the classic \textit{policy iteration} technique \cite{Puterman1994} does not apply to this risk-sensitive problem, which has a non-irreducible structure, and has a communicating class changing for different policies, and thus one needs to compare the cost over all the policies in order to find the optimal policy. In addition, since $|\mathbb{Y}|= \prod_{n=1}^N\left(\tau_n+1\right)$, the cardinality of the state-space is exponential in the number of clients, $N$.
 

However when the channel reliabilities are close to $1$, i.e., the system of interest is in a high-reliability asymptotic regime, and we are able to show that a simple ``modified-least-time-to-go" (MLG) policy, which is both structurally clean and easily implementable, is optimal. 
We note that the high-channel-reliability asymptotic is similar to the high SNR asymptotic in network information theory, see \cite{avestimehr2011wireless}. 

In this section, we derive some results useful for analyzing the MLG policy in the high-reliability regime. Later, in Section 8, we prove the optimality of the MLG policy in the case of two clients and propose a low-complexity policy for multi-clients which turns out to have good performance in the high-reliability regime.

For ease of exposition, we begin with the simple case of two clients sharing an AP. 
\subsection{Two-client Scenario and the MLG Policy}\label{se:MLG}
Consider two clients with channel reliabilities  $p_1=1-b_1\epsilon$, $p_2=1-b_2\epsilon$, where $\epsilon>0$ is a small quantity and $b_1,b_2>0$.
Suppose, without loss of generality, that $\tau_1=\tau$ and $\tau_2=\tau+\Delta$, where $\Delta \geq 0$. 

Define the \textit{modified-least-time-to-go (MLG)} policy by, 
\begin{align*}
f^\text{MLG}(\mathbf{x})=
\left\{
\begin{array}{ll}
2 \quad\quad  \mbox{if } \mathbf{x}=(0,\Delta-1),\\
\max_{n=1}^N \left\{ \arg \min_{n=1}^N \left(\tau_n-x_n \right) \right\} \quad \mbox{otherwise.}
\end{array}
\right.
\end{align*}
In words, the policy schedules the client with least $\tau_n-x_n$, i.e., ``least time to go", in most of the states, and breaks the ties by selecting the client with larger threshold. There is only one exception: When $\mathbf{x}=(0,\Delta-1)$, client $2$ is scheduled, although $\tau_2-x_2=\tau_1+1 > \tau_1-x_1=\tau_1$. 

We will show in Section \ref{se:Policy Design} that this MLG policy is indeed optimal when $\epsilon \rightarrow 0$.  We first explore its properties. 
\subsection{Regeneration Cycle}\label{se:regeneration cycle}
In the following, by the \textit{cost incurred in time slots $t_1,t_1 +1,\ldots,t_2$}, we mean the quantity, 
\begin{align}\label{eq:period cost}
\prod_{t=t_1}^{t_2}\exp\left(\theta \sum_{n=1}^{N} \mathds{1}\left\{Y_n(t)=\tau_n \right\}\right).
\end{align}
The \textit{regeneration point} of interest to us is defined as the time epoch when the system hits the state $(1,0)$, i.e., $Y(t)=(1,0)$. 

A \textit{regeneration cycle} is the time interval between two successive regeneration points. For any stationary policy $f$, 
let $v_{\text{cycle}}$ be the \textit{cost incurred in a regeneration cycle} (recall \eqref{eq:period cost}), and let $l_\text{cycle}$ be the \textit{length of the regeneration cycle}. Since $f$ is a stationary policy, $v_{\text{cycle}}$ and $l_\text{cycle}$ are random variables which are  i.i.d. in different regeneration cycles. Thus, as in renewal theory, we have, 
\begin{align}\label{eq:cycle cost to average cost}
\nonumber
J\left(f,\mathbf{x}\right)&=\lim_{T\rightarrow\infty} \frac{1}{\theta}\cdot \frac{1}{T}
\ln V_T^{f}(1,0)\\ \nonumber
&=\lim_{T\rightarrow\infty} \frac{1}{\theta} \frac{1}{T} \ln \mathrm{E}\left[\prod_{j=1}^{M_T^\text{(cycle)}}v_\text{cycle}^{(j)} \right]\\ \nonumber
&= \lim_{T\rightarrow\infty} \frac{1}{\theta} \frac{M_T^{(\text{cycle})}}{T} \ln \mathrm{E}\left[v_\text{cycle}\right]\\ &= \frac{1}{\theta} \frac{\ln \mathrm{E}\left[v_\text{cycle}\right]}{\mathrm{E}\left[l_\text{cycle} \right]}, \forall \mathbf{x}\in\mathbb{Y}.
\end{align}
where, the first equality follows from \eqref{eq:long term cost define}; since $v_\text{cycle}^{(j)}$ denotes the cost incurred during the $j$-th regeneration cycle, the second equality follows from the definition of $V_T^\pi(\mathbf{x})$ in \eqref{eq:cost define MDP2 V_T^Pi} with $M_T^{(\text{cycle})}$ denoting the total number of regeneration cycles during $T$ time slots, and the third equality holds since $v_\text{cycle}^{(j)},\forall j$ are i.i.d..

Result \eqref{eq:cycle cost to average cost}  reduces the analysis of the long-term average cost to the analysis of the expected cost and expected length of a regeneration cycle. Thus, it facilitates the following discussions. 

\subsection{SS-Point and SS-Period}\label{se:SS point}
Define the \textit{SS-points} as the time slots such that the packet-transmissions in the two successive time slots preceding this time-slot are both successful. See the examples in Fig.  \ref{fg:SS-points}. More formally, 
\begin{align}\label{eq:SS point define}
\tau^{ss}_j :& =
\begin{cases}
&\min\{t:t>0 \mbox{ and slots } t-1,t-2\mbox{ have}\\ 
& \quad\quad \mbox{ successful transmissions}\}   \mbox{ for } j=1\\ 
& \min\{t:t>\tau^{ss}_{j-1} \mbox{ and slots } t-1,t-2\mbox{ have}\\
& \quad\quad \mbox{ successful transmissions}\} \mbox{ for } j=2,3,\ldots
\end{cases}
\end{align}
Thus, $\tau^{ss}_j$ is the $j$-th SS-point. 
\begin{figure}[!t]
	\centering
	\includegraphics[width=0.5\textwidth]{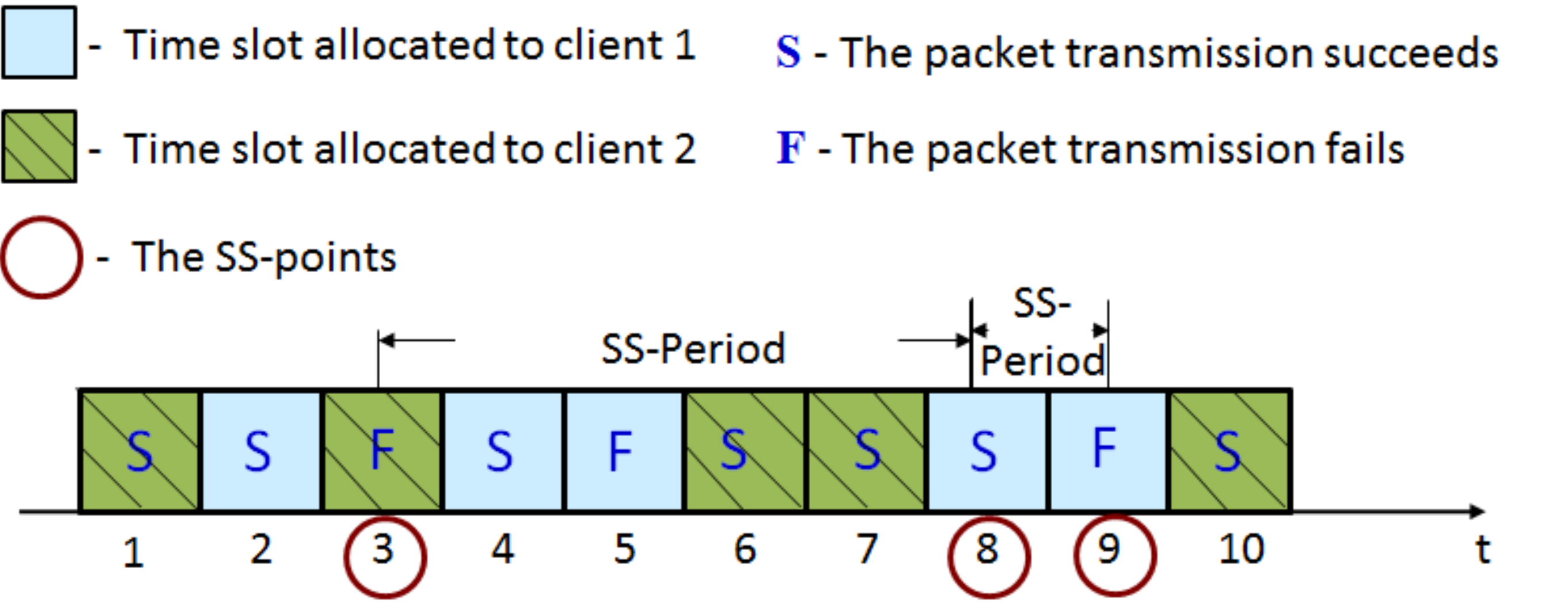}
	\caption{The SS-points and SS-periods are illustrated in a two-client scenario. (We arbitrarily allocate the time slots here just to give an example.)}
	\label{fg:SS-points}
\end{figure}

Define the 
time interval between two successive SS-points as an \textit{SS-period}, and let $v_\text{ss}(\mathbf{x})$ be \textit{the cost \eqref{eq:period cost} incurred during an SS-period} when the system state at the beginning of the SS-period is $\mathbf{x}$. Under the application of a stationary policy  the random variable $v_\text{ss}(\mathbf{x})$ is i.i.d. across different SS-periods for each fixed $\mathbf{x}$. 

It directly follows that, under an arbitrary NE policy (recall Appendix \ref{se:policies of interest}),  
\begin{align}\label{eq:ss temp 1}
\mathrm{P}\left(v_\text{ss}(\mathbf{x})>1 \right)=O(\epsilon),  
\end{align}
which is obtained by noting that: 
\begin{enumerate}[i)]
\item $v_\text{ss}(\mathbf{x})>1$ only if for some client $n$, we have $Y_n( t )=\tau_n$ for some time-slot $t$ in this SS-period.
\item However, $Y_n( t )=\tau_n$ for some time-slot $t$ in this SS-period is possible only if the length of this SS-period is $>1$ (which happens with probability $O(\epsilon)$).
\end{enumerate}
The statement i) follows the definition of $v_\text{ss}(\mathbf{x})$, while the statement ii) holds because NE policies do not serve client $2$ when the system state is $(\tau_1,0)$, and do not serve client $1$ when the state is $(0,\tau_2)$, and thus any possible starting state $\mathbf{x}$ of an SS-period satisfies $x_n<\tau_n,n=1,2$.

Further, if $P_{\mathbf{x}}^{(\text{ss})}$ is the probability that in a regeneration cycle (i.e., time between two successive hits of the state $(1,0)$, recall Section \ref{se:regeneration cycle}), there is at least one SS-period in which the starting state is $\mathbf{x}$, then,
\begin{align}\label{eq:ss temp 2}
\mathrm{E}\left[v_\text{cycle} \right]&=1+\sum_{ \mathbf{x} }P_{\mathbf{x}}^{(\text{ss})}\left(\mathrm{E}\left[v_\text{ss}\left(\mathbf{x}  \right) \right]-1\right)+o(\epsilon^k),
\end{align}
where $k\geq 1$ is an integer such that
$d_1 \epsilon^k \leq \mathrm{E}\left[v_\text{cycle} \right]-1 \leq d_2 \epsilon^k  $ for some $d_1,d_2>0$ when $\epsilon$ is sufficient small. This is obtained by noting that, 
\begin{enumerate} 
\item[1)] A regeneration cycle consists only of SS-periods. 
\item[2)] For any regeneration cycle, $v_\text{cycle}>1$ only if at least one of the SS-periods included in this cycle has    $v_\text{ss}(\cdot)>1$. 
\item[3)] The probability that two or more SS-periods incur a cost $v_\text{ss}(\cdot)>1$ is much less than the probability that only one of these SS-periods incurs a cost $v_\text{ss}(\cdot)>1$, when $\epsilon$ is small enough. The probability of this event being small follows \eqref{eq:ss temp 1}.
\end{enumerate}
In the above, the proof of statements $1$) and $2$) is direct and omitted.
It should be noted that the technique of ignoring events having relatively small probabilities, as employed in the proofs of statement $3$) and \eqref{eq:ss temp 2}, is frequently used in the remaining of this paper. 
These results facilitate our following analyses.

Now, we consider a \textit{regeneration cycle consisting of only successful transmissions}. Denote $\mathbb{X}_\text{ss}$ as the set of all the system states hit during such a regeneration cycle. 
\begin{lemma}\label{th:ss Lemma}
The following result holds:  
\begin{enumerate}[1)]
\item For an arbitrary NE policy, 
\begin{align}\label{eq:ss cost to cycle cost inequality}
\mathrm{E}\left[v_\text{cycle} \right]\geq 1 + \sum_{ \mathbf{x} \in \mathbb{X}_{ss} }\left(\mathrm{E}\left[v_\text{ss}\left(\mathbf{x}  \right) \right]-1\right)+o(\epsilon^k),
\end{align}
where $k$ is an integer such that  $d_1 \epsilon^k \leq \mathrm{E}\left[v_\text{cycle} \right]-1 \leq d_2 \epsilon^k  $ for some $d_1,d_2>0$ when $\epsilon$ is sufficient small.
\item When the MLG policy is applied, 
\begin{align}\label{eq:ss cost to cycle cost marginal switching}
\mathrm{E}\left[v_\text{cycle} \right]= 1 + \sum_{ \mathbf{x} \in \mathbb{X}_{ss} }\left(\mathrm{E}\left[v_\text{ss}\left(\mathbf{x}  \right) \right]-1\right)+o(\epsilon^k),
\end{align}
where $k$ is similarly defined as for \eqref{eq:ss cost to cycle cost inequality}. 
\item 
Further, when the MLG policy is applied, and if  $\Delta \geq 2$, 
\begin{align}
\label{eq:MLG Xss set}
& \mathbb{X}_\text{ss}=\{(1,0), (0,x_2),  \forall x_2=0,\cdots,\Delta\!-\!1 \}
\\ \label{eq:MLG cycle length}
&\mathrm{E}[l_\text{cycle}] = \Delta +O(\epsilon); 
\\ \label{eq:MLG ss period cost 1}
&\mathrm{E}[v_\text{ss}(1,0)]=1+b_1^{\tau-1}\epsilon ^{\tau-1}\left(\exp(\theta)-1 \right)+O(\epsilon^{\tau}), 
\\ \label{eq:MLG ss period cost 2}
&\mathrm{E}[v_\text{ss}(0,x_2)]=1+O(\epsilon^{\tau}), \forall x_2 = 0,\cdots,\Delta-1.
\end{align}
Similar results can be obtained for $\Delta=0,1$.
\end{enumerate}
\end{lemma}
\begin{proof}
These results follow from \eqref{eq:ss temp 1}, \eqref{eq:ss temp 2}, and the definition of $\mathbb{X}_\text{ss}$. 
The proof is straightforward, and the details are omitted.
\end{proof}

\section{\!\!\! Asymptotically Optimal Policies} \label{se:Policy Design}

\subsection{The Two-Client Scenario}\label{se:2-client scenario}
Consider the case where there are two-clients sharing an AP, with channel reliabilities $p_1=1-b_1\epsilon$, $p_2=1-b_2\epsilon$, and thresholds $\tau_1=\tau$, $\tau_2=\tau+\Delta$. Without loss of generality assume  $\Delta \geq 0$. Recall the definition of the MLG policy in Section \ref{se:MLG}, and note that $b_1,b_2,\epsilon>0$. The following theorem establishes the optimality of the MLG policy.

\begin{theorem}\label{th:MLG policy}
The following results hold:
\begin{enumerate}[1)]
\item The risk-sensitive cost under MLG policy  is, $\forall \mathbf{x}$,
\begin{align*}
J(f^\text{MLG},\mathbf{x})=
\begin{cases}
 A_0 \epsilon^{\tau-1}+O(\epsilon^{\tau})
& \mbox{if } \Delta=0 \\
\frac{\mathrm{e}^\theta\!-\!1  }{2\theta}\epsilon^{\tau-1} \sum_{j=0}^{\tau\!-\!1}b_1^j b_2^{\tau\!-\!1\!-\!j} \!+\!O(\epsilon^\tau\!)\! & \mbox{if } \Delta =1 \\
\frac{\mathrm{e}^\theta-1 }{\theta \Delta}b_1^{\tau-1} \epsilon^{\tau-1}+O(\epsilon^{\tau}) & \mbox{if } \Delta \geq 2 \\
\end{cases}
\end{align*}
where
\begin{align*}
A_0=  \frac{\mathrm{e}^\theta-1}{\theta} \sum_{j=1}^{\tau-2}b_1^j b_2^{\tau-1-j}
+\frac{b_1^{\tau\!-\!1}+b_2^{\tau-1}}{2\theta} \left(\mathrm{e}^{2}-1 \right).
\end{align*}

\item \!The optimal cost $J^\star(\mathbf{x})\!:=\min_f \! J(f,\mathbf{x})$ has a lower bound,  
\begin{align*}
J^\star(\mathbf{x})\geq
\left\{
\begin{array}{ll}
A_0 \epsilon^{\tau-1}+o(\epsilon^{\tau-1}) & \mbox{if } \Delta =0 \\
\frac{\mathrm{e}^\theta-1}{2\theta} A_1  \epsilon^{\tau-1}+o(\epsilon^{\tau-1}) & \mbox{if } \Delta=1 \\
\frac{\mathrm{e}^\theta-1  }{\theta} A_2 \epsilon^{\tau-1} +o(\epsilon^{\tau-1})& \mbox{if } \Delta \geq 2
\end{array}
\right.
\end{align*}
where
\begin{align*}
A_0 & \mbox{ is as in the  statement above} \\
b_{\min}&:=\min\{b_1,b_2\} \\
A_1&:=b_1^{\tau-1}+\left(\tau-1\right)b_{\min}^{\tau-1}\\
A_2 &:=\min \bigg\{\frac{b_1^{\tau-1}}{\Delta},  \frac{b_1^{\tau-1}+(\tau_1-1)b_{\min}^{\tau-1}}{\Delta+1}, \\
& \quad
\frac{b_1^{\tau-1}+(\tau_1\!-\!1)b_{\min}^{\tau-1}+\sum_{j=1}^{\tau-1}b_2^j b_1^{\tau-1-j}}{\Delta+2}
\bigg\}.
\end{align*}
\item \!\!Thus, it follows from $1$) and $2$) above that the MLG policy is optimal in the high reliability asymptotic regime (i.e., small $\epsilon$) if any of the following conditions is satisfied:
\begin{enumerate}
\item[(i)] $\Delta=0$; \quad \quad
\quad \quad \quad \quad (ii) $\Delta=1$, and $b_1\leq b_2$;
\item[(iii)] $\Delta\geq 2$, and
$b_1^{\tau-1}\leq \Delta (\tau-1)b_2^{\tau-1}$.
\end{enumerate}
\end{enumerate}
\end{theorem}

\begin{proof}
We will only consider the case when  $\Delta\geq 2$, since the analyses for the cases when $\Delta=1$ or $0$ follows similar arguments.

By Lemma \ref{th:ss Lemma}, under the application of the MLG policy, we have,
\begin{align}\label{eq:MLG3}
\mathrm{E}[v_\text{cycle}]=1+b_1^{\tau-1}\epsilon ^{\tau-1}\big(\exp(\theta)-1 \big)+O(\epsilon^{\tau}).
\end{align}
Thus, statement 1) follows by combining \eqref{eq:cycle cost to average cost}, \eqref{eq:MLG cycle length} and \eqref{eq:MLG3}.

To prove statement 2), we begin by deriving the lower bound of $\mathrm{E}[v_\text{ss}(\mathbf{x})]$ for any system state $\mathbf{x}$ of the form $(\cdot,0)$ or $(0,\cdot)$. Note that, any possible starting state of an SS-period is of this form. In the following, we focus on the analysis of $\mathrm{E}[v_\text{ss}(1,0)]$, since the analysis of the cost for the SS-period starting with any other state follows similar arguments. 

Consider the evolution of system over an SS-period starting with state $(1,0)$ under the application of an arbitrary stationary policy. Then we have the following two possibilities, 
\begin{enumerate}[(a)]
\item The policy serves client $2$ before the earlier of these two events: $i$) a successful packet delivery for client $1$, $ii$) the system hits the value $(\tau,\tau-1)$. 
Under such a policy, it can be shown that a cost $v_\text{ss}(1,0)>1$ is incurred with a probability $>d \epsilon^{\tau-2}$ for some $d>0$.
\item The policy does not serve client $2$ before the earlier of the following two events: $i$) a successful packet delivery for client $1$, $ii$) the system hits the value $(\tau,\tau-1)$. 
Then, if failures occur in all of the first $\tau-1$ time slots for the SS-period, the state $(\tau,\tau-1)$ will be hit. Thus, a cost  $v_\text{ss}(1,0)\geq \exp(\theta)$ is incurred with a probability $ \geq b_1^{\tau-1}\epsilon^{\tau-1}$.

\end{enumerate}
Consequently it follows from a) and b) above,  $$\mathrm{E}[v_\text{ss}(1,0)]\geq 1+b_1^{\tau-1}\epsilon^{\tau-1}+o(\epsilon^{\tau-1}).$$
Similar arguments lead us to conclude the following lower bounds on $\mathrm{E}[v_\text{ss}(\mathbf{x})]$ under the application of an arbitrary stationary policy, (recall that $\mathbf{x}$ should be of the form $(0,\cdot)$ or $(\cdot,0)$ since it is a possible starting state of an SS-period)
\begin{enumerate} [i.]
\item $\forall \mathbf{x}\in \left\{(0,x_2)|x_2\leq \Delta-1 \right\}$, $ \mathrm{E}[v_\text{ss}(\mathbf{x}) ]\geq 1$.
\item $\forall \mathbf{x}\in \left\{(0,x_2)|x_2\geq  \Delta+2 \right\}\bigcup \left\{(x_1,0)|x_1 \geq 2 \right\}$, $ \mathrm{E}[v_\text{ss}(\mathbf{x}) ] \geq 1+d \epsilon^{\tau-2}+o(\epsilon^{\tau-2})$, with some $d>0$.
\item $\mathrm{E}[v_\text{ss} (1,0) ]\geq 1+b_1^{\tau-1}\epsilon ^{\tau-1}\left(\exp(\theta)-1 \right)+o(\epsilon^{\tau-1})$.
\item  $\mathrm{E}[v_\text{ss}(0,\Delta+1) ]\geq 1+\sum_{j=1}^{\tau-1}b_2^j b_1^{\tau-1-j}\epsilon^{\tau-1}+o(\epsilon^{\tau-1})$.
\item   $\mathrm{E}[v_\text{ss}(0,\Delta)] \geq 1+ (\tau-1)b_{\min}^{\tau-1}\epsilon^{\tau-1}+o(\epsilon^{\tau-1})$.
\end{enumerate}
By combining these results with the inequality \eqref{eq:ss cost to cycle cost inequality} and using equation \eqref{eq:cycle cost to average cost}, we obtain the second statement.

The third statement is a simple consequence of the first two statements.
\end{proof}

In Theorem \ref{th:MLG policy}, the first statement characterizes the risk-sensitive cost of the MLG policy, while the second provides a lower bound on the cost for any stationary policy. The third statement provides three sufficient conditions under which the MLG policy is asymptotically optimal. 
These conditions are related to the difference between the inter-delivery thresholds for different clients, and the difference in their relative failure probabilities.

\subsection{The General Case: N Clients in the High-Reliability Regime}\label{se:multi-client subsection}
Now, we consider the general case where there are $N$ clients sharing an AP, with the channel reliability of the $n$-th client being $p_n=1-b_n \epsilon$, where $\epsilon>0$ is a small quantity and $b_n>0$. The inter-delivery threshold of client $n$ is $\tau_n$. It is assumed that $N\leq \tau_1\leq \tau_2 \leq \cdots \leq \tau_N$, with out loss of generality. 

Since Theorem \ref{th:optimality equation} implies that there exists a stationary policy that is optimal for the MDP-2, we focus exclusively on stationary policies. Now, we obtain characterization of the optimal policy. 

We define a \textit{regeneration point} as the time epoch when the system hits the state $(0,1,\cdots,N-1)$, i.e., time $t$ is a regeneration point iff. $Y(t)=(0,1,\cdots,N-1)$. 
(Recall that for the case of $2$ clients as discussed in Section \ref{se:regeneration cycle}, the regeneration point is the epoch when the state $(1,0)$ is hit.) 
The 
\textit{regeneration cycle} 
is the time interval between two successive regeneration points. 
Now, consider a \textit{regeneration cycle consisting of only successful transmissions}, and denote by $\mathbb{X}_\text{sN}$ the sequence of system states hit during such a regeneration cycle. Note that $\mathbb{X}_\text{sN}$ is a deterministic sequence with a given stationary policy.
Let $|\mathbb{X}_\text{sN}|$ be the length of this sequence, and $\mathbb{X}_\text{sN}(j)$ be the $j$-th state in this sequence.

We also define an \textit{SN-point} as the time slot when the packet transmissions in the $N$ successive time slots preceding  this time-slot 
are all successful. (This is similar to the definition of SS-point, in \eqref{eq:SS point define}.)
The \textit{SN-period} is the time interval between two successive SN-points, and let $v_\text{sN}(\mathbf{x})$ denote the cost \eqref{eq:period cost} incurred during an SN-period when the system state at the beginning of the SN-period is $\mathbf{x}$, similar to the two-client case. 

We further consider a time interval comprising of no less than $N$ time-slots, which starts when the system state assumes the value $\mathbf{x}$ and ends when the nearest SN-point (such that the length of the period $\geq N$) is hit. Denote by  $\tilde{v}_\text{sN}(\mathbf{x})$ the cost \eqref{eq:period cost} incurred during such a period. 
(Note that this is different from ${v}_\text{sN}(\mathbf{x})$ because  ${v}_\text{sN}(\mathbf{x})$ is the cost incurred during an SN-period, and that an SN-period may have a length strictly less than $N$.  An example of SN-period with length $1$ is shown in Fig. \ref{fg:SS-points} for the two-client scenario.)

\begin{lemma}\label{th:MC lemma}
	The optimal policy is a member of the set,  
	\begin{align*}
	\{f: J(f,\mathbf{x})=O(\epsilon), \forall  \mathbf{x}\in\mathbb{Y}\}.  
	\end{align*}
	Further, for any policy $f$ in this set, the following results hold: 
	\begin{enumerate}[1)]
		\item  The risk-sensitive cost is,  
		\begin{align}\label{eq:MC risk sensitive cost form}
		J(f,\mathbf{x})=\frac{1}{|\mathbb{X}_\text{sN}|}\sum_{ j=1}^{|\mathbb{X}_\text{sN}|}\left(\mathrm{E}\left[v_\text{sN}\big(\mathbb{X}_\text{sN}(j) \big) \right]\!-\!1\right)+o(\epsilon^k), 
		\end{align}
		for any system state $\mathbf{x}$, where $k\geq 1$ is an integer such that  $d_1 \epsilon^k \leq J(f,\mathbf{x}) \leq d_2 \epsilon^k  $ for some $d_1,d_2>0$ when $\epsilon$ is sufficient small.
		\item For any possible starting state $\mathbf{x}$ of an SN-period, we have, 
		\begin{align}\label{eq:MC tilta form}
		\mathrm{E}[\tilde{v}_\text{sN}(\mathbf{x})]=1+ \sum_{j=0}^{N-1} \left(\mathrm{E}\left[v_\text{sN} ({\mathcal{S}}^j(\mathbf{x})) \right]-1 \right)+o(\epsilon^k), 
		\end{align}
		where $\mathcal{S}^1(\mathbf{x})$ is the state that succeeds state $\mathbf{x}$ in the event of a successful transmission when  policy $f$ is applied, i.e., 
		\begin{align*}
		&\mathcal{S}^1\!(\mathbf{x}\!)\!:=\!(x_1\!+\!1,\!\cdots\!, x_{\!f(\mathbf{x})\!-\!1}\!+\!1,0,x_{\!f(\mathbf{x})\!+\!1}\!+\!1,\!\cdots\!,x_{\!N}\!+\!1) \!\wedge \!\boldsymbol{\tau}; \\
		&\mbox{also } \mathcal{S}^{j+1}(\mathbf{x}):=\mathcal{S}\left(\mathcal{S}^j(\mathbf{x})\right),j=1,2,\cdots ; \mathcal{S}^0(\mathbf{x}):=\mathbf{x}, 
		\end{align*}
		and $k$ is an integer such that  $d_1 \epsilon^k \leq \mathrm{E}\left[\tilde{v}_\text{sN}(\mathbf{x}) \right]-1 \leq d_2 \epsilon^k  $ for some $d_1,d_2>0$ when $\epsilon$ is sufficient small.
	\end{enumerate}
\end{lemma}
\begin{proof}
	The results are obtained using arguments similar to the case of two-clients in the high reliability regime (See Section \ref{se:regeneration cycle}, \ref{se:SS point}, equations \eqref{eq:cycle cost to average cost}
	\eqref{eq:ss temp 1}
	\eqref{eq:ss temp 2}
	\eqref{eq:ss cost to cycle cost inequality}
	).
\end{proof}
One may note that the r.h.s. of \eqref{eq:MC tilta form} is closely related to the r.h.s. of \eqref{eq:MC risk sensitive cost form}, by noting that $\mathcal{S}^j\big(\mathbb{X}_\text{sN}(1)\big)=\mathbb{X}_\text{sN}(j+1),\forall j=1,\cdots,|\mathbb{X}_\text{sN}|\!-\!1$, and that $|\mathbb{X}_\text{sN}|\geq N$ holds whatever policy is applied. 
Thus, the following assumption is not restrictive.

\begin{assumption}\label{as:two target equal}
A stationary policy that minimizes $\mathrm{E}[\tilde{v}_\text{sN}(\mathbf{x})]$ for each system state $\mathrm{x}\in\mathbb{Y}$, also minimizes $J(f,\mathbf{x})$.
\end{assumption}

\IncMargin{1em}
\begin{algorithm} 
	\caption{SN Policy Algorithm}\label{ag:whole}
	\DontPrintSemicolon
	\SetKwData{Left}{left}\SetKwData{This}{this}\SetKwData{Up}{up}
	\SetKwFunction{Union}{Union}\SetKwFunction{FindCompress}{FindCompress}
	\SetKwInOut{Input}{input}\SetKwInOut{Output}{output}
	\Input{$N$, $\theta$, $\tau_1,\cdots,\tau_N$,   $b_1,\cdots,b_N$. }
	\Output{Policy $g(\mathbf{x}), \forall \mathbf{x}\in \mathbb{Y}$.}
	\BlankLine
	$\mathbb{Y}_0=\{\mathbf{x}: \exists A>0,  \min_{\pi}\left(\mathrm{E}[\tilde{v}_\text{sN}(\mathbf{x})]\right)=1+A + o(1)\};$ \label{ag:line 1}
	\;
	\ForEach{$\mathbf{x}\in\mathbb{Y}_0$}
	{$A(\mathbf{x})$ is as in Step \ref{ag:line 1};
		\;
		$g(\mathbf{x})\leftarrow \arg\min_{n=1}^N A(\tilde{\mathcal{S}}_n(\mathbf{x}))$	;}
	$\mathbb{Z}\leftarrow \emptyset$;
	$\mathbb{Y}_\text{remain}\leftarrow \emptyset$\;
	\ForEach {$k=1$ to $(\min_{n=1}^N \tau_n)$}
	{$\mathbb{Z}\leftarrow \mathbb{Z}\cup \mathbb{Y}_{k-1}$;\;
		$\mathbb{Y}_k\leftarrow\{\mathbf{x}:\mathbf{x+1}\in \mathbb{Y}_{k-1} \text{ and } \mathbf{x} \notin \mathbb{Z} \}$;\;
		\lRepeat{$Y_k$ not extend;}{$\mathbb{Y}_k\leftarrow \mathbb{Y}_k \cup \{\mathbf{x}: \tilde{\mathcal{S}}_n(\mathbf{x})\in \mathbb{Z}\cup \mathbb{Y}_k,\forall n\text{ and } \mathbf{x}\notin \mathbb{Z}\cup \mathbb{Y}_k \} $}}
	\ForEach {$k=1$ to $(\min_{n=1}^N \tau_n)$}
	{ $\mathbb{Y}_k^\prime\leftarrow \mathbb{Y}_k;$\;
		\Repeat{$\mathbb{Y}_k^\prime=\emptyset$ or $\mathbb{Y}_k^\prime=\mathbb{Y}_k^{\prime\prime};$}
		{
			$\mathbb{Y}_k^{\prime\prime}\leftarrow \mathbb{Y}_k^\prime;$ 
			$\mathbb{Y}_k^\prime\leftarrow \emptyset;$
			\;	
			\ForEach {$\mathbf{x}\in \mathbb{Y}_k^{\prime\prime}$}
			{$m\leftarrow \max \{j: \exists n, \tilde{\mathcal{S}}_n(\mathbf{x})\in \mathbb{Y}_j \}$;\;
				$U_\text{set}\leftarrow \{n: \tilde{\mathcal{S}}_n(\mathbf{x})\in \mathbb{Y}_m \}$;\;
				\If{m>k}{$[A\!(\mathbf{x}),\!g(\mathbf{x})]\!\leftarrow\! \min_{n\in U_\text{set}}\!\! b_n A((\mathbf{x}\!+\!\boldsymbol{1})\!\wedge \!\boldsymbol{\tau})$;  }
				\ElseIf{$\exists n\!\in\! U_\text{set}$,  $A(\tilde{S}_n(\mathbf{X}))$ not yet}{$\mathbb{Y}_k^\prime\leftarrow \mathbb{Y}_k^\prime \cup \mathbf{x};$}
				\Else{
					$[A(\mathbf{x}),g(\mathbf{x})] \leftarrow \min_{n\in U_\text{set}}A(\tilde{\mathcal{S}}_n(\mathbf{x})) +b_n A((\mathbf{x}+\boldsymbol{1})\wedge \boldsymbol{\tau})\mathds{1}\!\{(\mathbf{x+1})\!\wedge\!\boldsymbol{\tau}\!\in\!\mathbb{Y}_{k\!-\!1} \}$;}
			}
		} 
		\If{$\mathbb{Y}_k^\prime\neq \emptyset$}
		{
			$\mathbb{Y}_\text{remain}\leftarrow\mathbb{Y}_\text{remain}\cup\mathbb{Y}_k^\prime;$
			$B(\mathbf{x})\leftarrow0,\forall \mathbf{x}\in \mathbb{Y}_k;$
			\;
			\ForEach{$n=1$ to $N-1$}
			{
				\ForEach{$\mathbf{x}\in \mathbb{Y}_k$}
				{ 	
					$U_\text{set}^\prime\leftarrow g(\mathbf{x}) \text{ or } U_\text{set}$;\; 
					$B(\mathbf{x})\leftarrow\min_{n\in {U_\text{set}^\prime}} B(\tilde{S}_n(\mathbf{x}))+ b_n A((\mathbf{x+1})\!\wedge\!\boldsymbol{\tau}) \mathds{1}\!\{(\mathbf{x+1})\!\wedge\!\boldsymbol{\tau}\in\mathbb{Y}_{k-1} \}$
				}
			}
			
			\ForEach{$\mathbf{x}\in \mathbb{Y}_k^\prime$}
			{
				$[A(\mathbf{x}),g(\mathbf{x})] \leftarrow \min_{n\in U_\text{set}^\prime}B(\tilde{\mathcal{S}}_n(\mathbf{x})) +b_n A((\mathbf{x}+\boldsymbol{1})\wedge \boldsymbol{\tau})\mathds{1}\!\{(\mathbf{x+1})\!\wedge\!\boldsymbol{\tau}\!\in\!\mathbb{Y}_{k\!-\!1} \}$
			}
			
		}
	}
\end{algorithm}
\IncMargin{1em}

Consequently, we design Algorithm  \ref{ag:whole} above to obtain a stationary policy, denoted  \textit{SN policy}, which tends to minimizes $\mathrm{E}[\tilde{v}_\text{sN}(\mathbf{x})]$ for any system state $\mathbf{x}\in \mathbb{Y}$. 
Here, 
\begin{align*}
\tilde{\mathcal{S}}_n(\mathbf{x}\!):=(x_1\!+\!1,\cdots,x_{n\!-\!1}\!+\!1,0,x_{n\!+\!1}\!+\!1,\cdots \!,x_N\!+\!1)\wedge \boldsymbol{\tau},  
\end{align*}
i.e., the state that succeeds the state $\mathbf{x}$ in the event of a successful transmission for client $n$ in the MDP-2. 
This algorithm divides the state space into sets, \\ $\mathbb{Y}_0,\mathbb{Y}_1,\cdots,\mathbb{Y}_{\min_{n=1}^N\{\tau_n\}}$, with,
\begin{align}\label{eq:algo explain}
\mathbb{Y}_k\!:=\!\{\mathbf{x}: \exists A>0, \min_{\pi}\left(\mathrm{E}[\tilde{v}_\text{sN}(\mathbf{x})]\right)\!=\!1\!+\!A \epsilon^k \!+\! o(\epsilon^k) \} ,
\end{align}
for $k=0,1,\cdots, \min_{n=1}^N\{\tau_n\}$.
Then it obtains or approximates the $A$ in \eqref{eq:algo explain} for each system state $\mathbf{x}$ (denoted $A(\mathbf{x})$), and decides the optimal control based on $A(\mathbf{x}) $.

\begin{theorem}\label{th:SN theorem}
When Assumption \ref{as:two target equal} holds, and $\mathbb{Y}_\text{remain}$ in Algorithm \ref{ag:whole} is empty, then the SN policy is optimal in the high reliability asymptotic regime.
\end{theorem}

Theorem \ref{th:SN theorem} directly follows from Assumption \ref{as:two target equal} and the design of Algorithm \ref{ag:whole}.\footnote{Note that the Algorithm \ref{ag:whole} can be further improved by using $\tilde{S}(\tilde{S}(\mathbf{x}))$ in Step 19-22.	}  The example in Fig. \ref{fg:case_3} illustrates the simulation for a multi-client system when these asymptotic conditions are satisfied.

\section{Simulations}
We now present the results of a simulation study comparing several wireless scheduling policies with respect to their risk-sensitive average costs. We present the results for the scenarios with clients requiring different inter-delivery thresholds and under heterogeneous channel reliabilities.

The wireless scheduling policies implemented include the optimal policy (OP) obtained from Theorem \ref{th:optimality equation}, \quad the modified-least-time-to-go (MLG) policy (for two-client scenario) proposed in Section \ref{se:2-client scenario}, and 
the SN policy proposed in Section \ref{se:multi-client subsection}. Also, two other heuristic policies are compared: 
the packet-level round-robin policy (PRR),  and the largest-weighted-delivery-debt (WDD) policy, which serves the client with the largest weighted delivery debt, where:
\begin{align*}
\textrm{Delivery Debt}_n=\frac{t}{p_n \tau_n}-\frac{M_t^{(n)}}{p_n}.
\end{align*}
(Recall $M_t^{(n)}$ is the number of packets delivered for the $n$-th client by time $t$, as in \eqref{eq:cost define origninal}.)
The WDD policy has been known to be ``timely-throughput" optimal (see \cite{Hou2009} for discussion).


Fig. \ref{fg:case_1} shows the costs incurred by these four wireless scheduling policies for different risk-sensitive parameters. It can be seen that the optimal policy always outperforms all the other policies. 

Fig. \ref{fg:case_2b} compares the scheduling policies under different channel reliabilities in the two-client scenario. It can be seen that even when the channel reliability  probabilities are only moderate, e.g., $p_1=0.6$ and $p_2=0.8$, the MLG policy still achieves almost the optimal cost, and outperforms all other greedy policies. 

Fig. \ref{fg:case_3} compares the scheduling policies in a  multi-client scenario. It can be seen that even when the channel reliability  probabilities are only moderate, e.g. $0.8$, SN policy still approximates the optimal cost, and outperforms all other greedy policies. 
Here, we also employ the periodic scheduling (PS) policy\cite{Bar-Noy1998}, which is optimal when the failure probabilities are exactly zero. 
It can be seen that the PS policy performs extremely poorly even when the failure probability is very small, e.g., $0.01$, since it gives rise to open-loop policies.
In contrast, the high-reliability asymptotic approach proposed for the scenario with sufficiently-small failure probability provides a well performing closed-loop scheduling policy.
This confirms the value of the high-reliability asymptotic approach.

\begin{figure}[!t]
	\centering
	\includegraphics[width=0.5\textwidth]{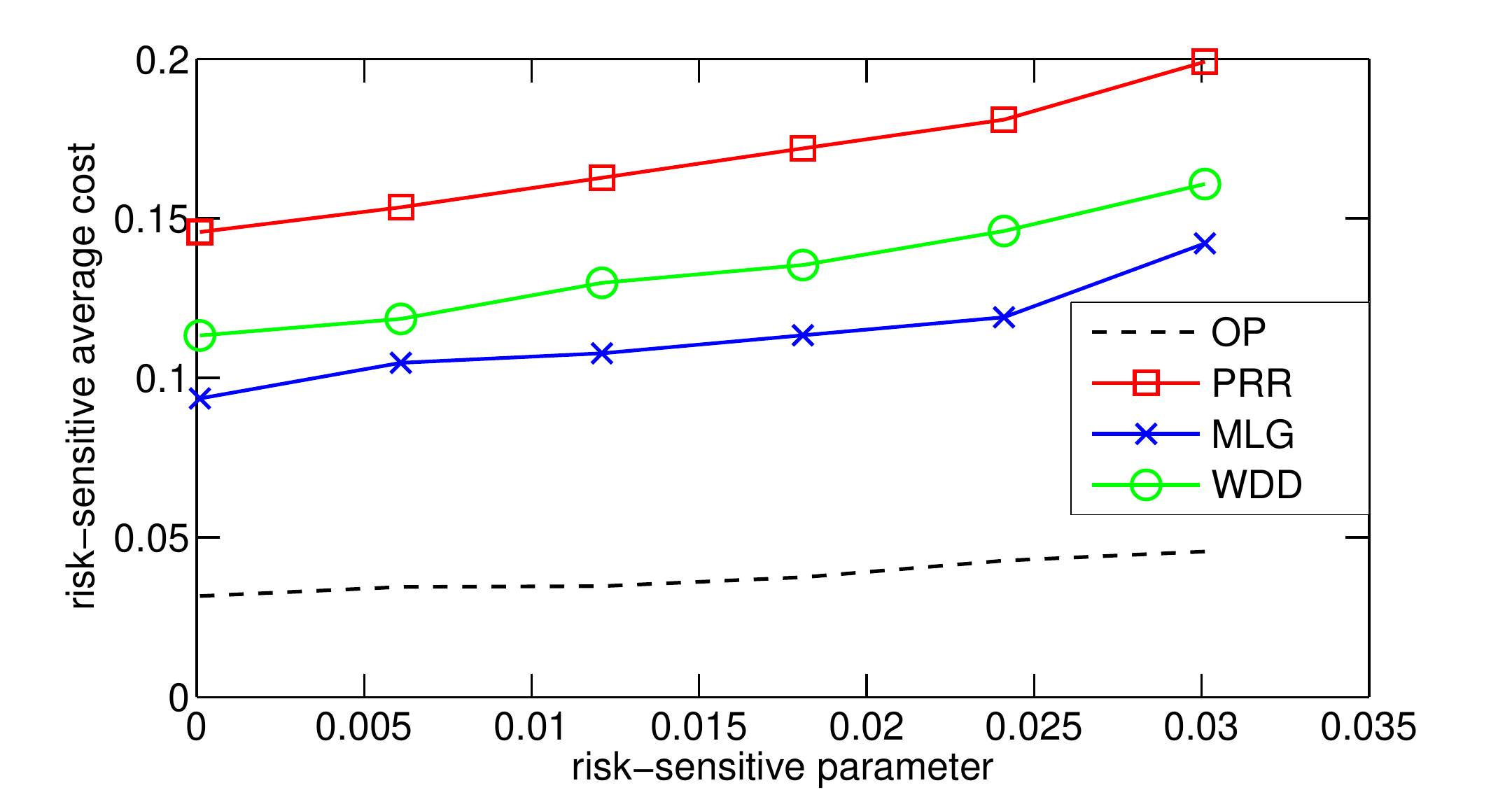}
	\caption{The risk-sensitive average cost vs. the risk-sensitive parameter $\theta$ for different wireless scheduling policies is shown. (The parameters are $N=2$, $p_1=0.4$, $p_2=0.1$, $\tau_1=20$, $\tau_2=40$).}
	\label{fg:case_1}
\end{figure}

\begin{figure}[!t]
	\centering
	\includegraphics[width=0.5\textwidth]{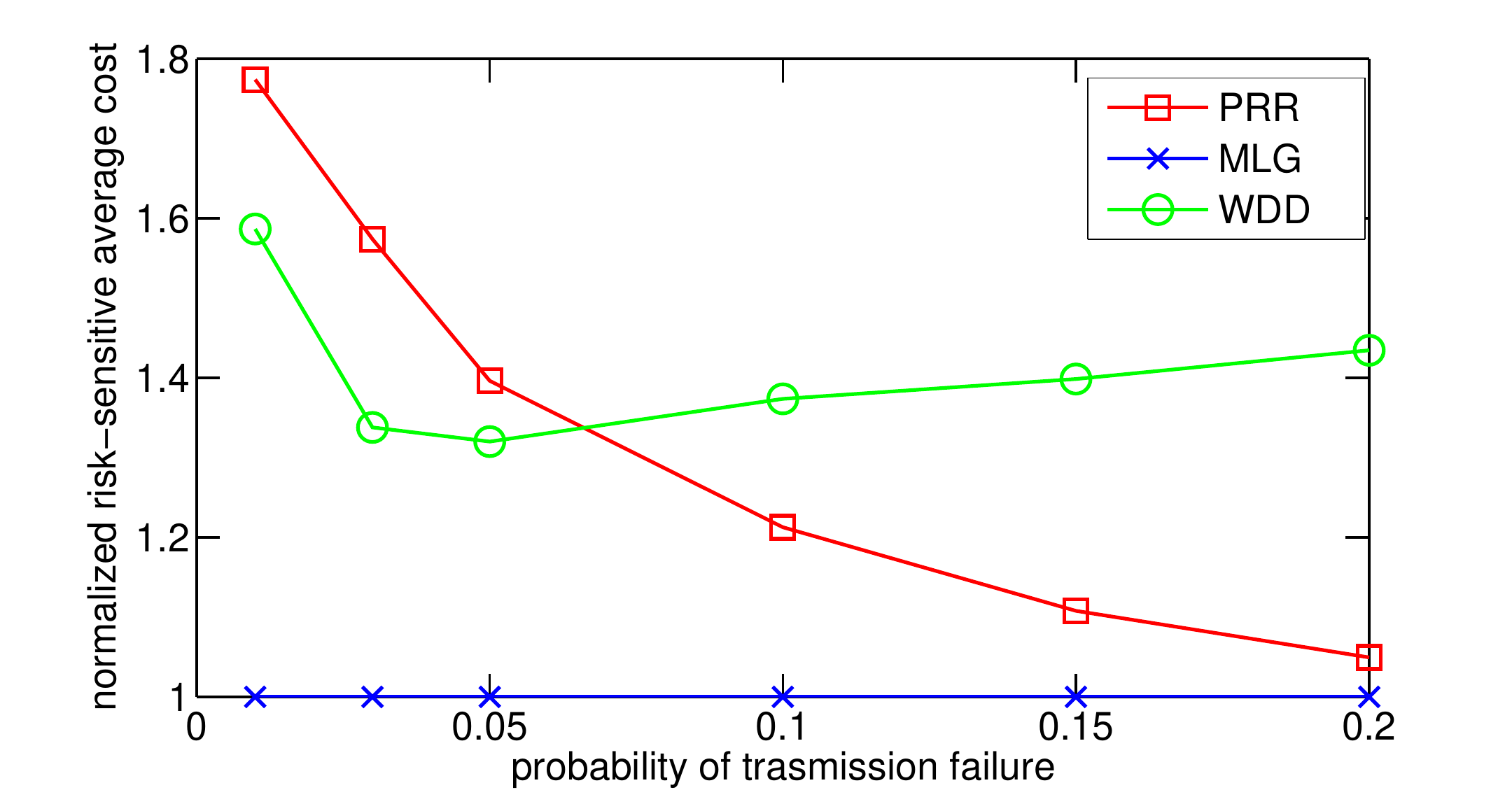}
	\caption{In two-client scenario, the normalized risk-sensitive average cost (normalized by the cost of the optimal policy) vs. the failure transmission parameter $\epsilon$. ($p_1=1-2\epsilon$, $p_2=1-\epsilon$, $\tau_1=3$, $\tau_2=5$, $\theta=0.01$.)}
	\label{fg:case_2b}
\end{figure}

\begin{figure}[!t]
	\centering
	\subfigure[Performance of the PRR, WDD, and SN policies] {\includegraphics[width=0.5\textwidth]{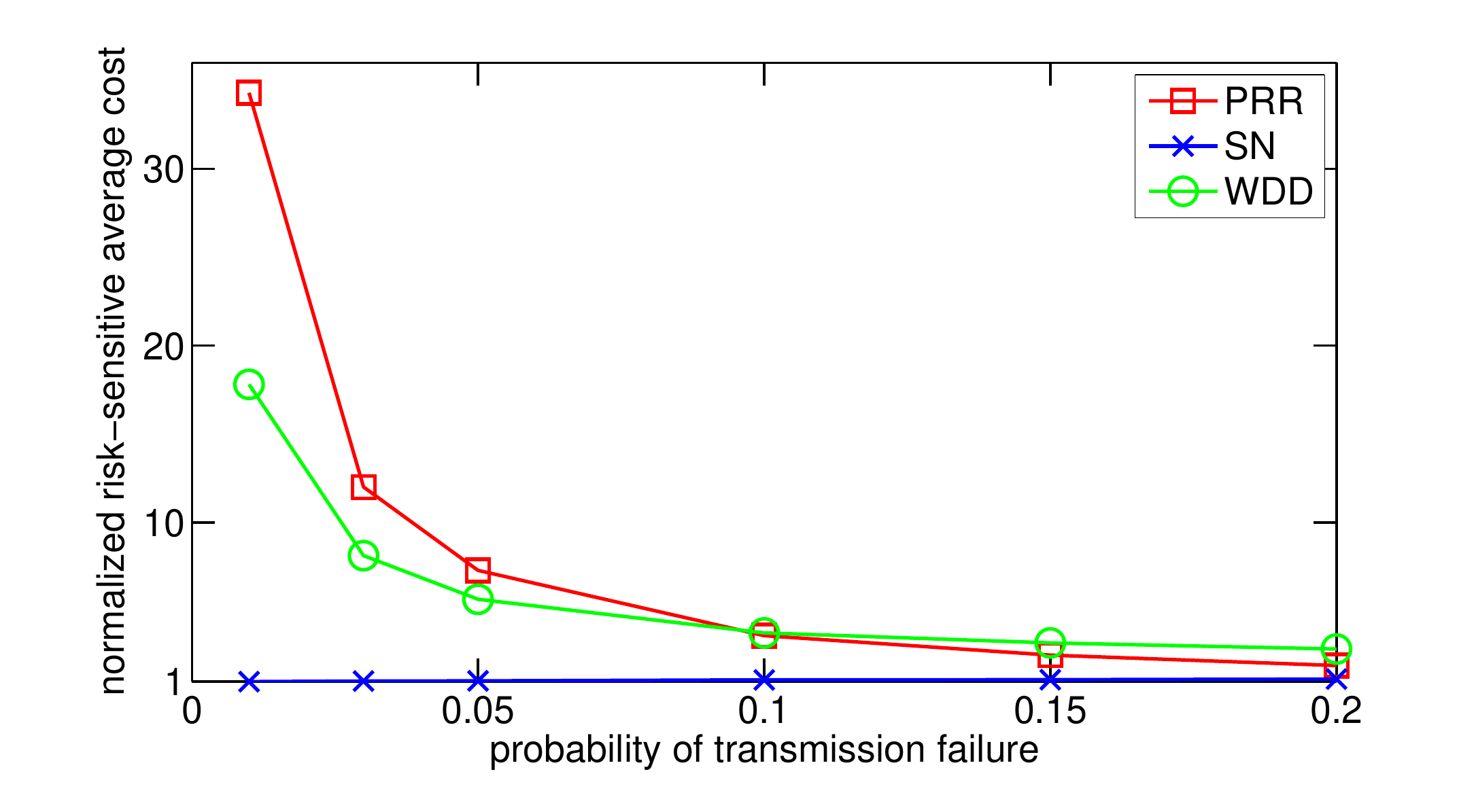}}
	\subfigure[Performance of Periodic Scheduling] {\includegraphics[width=0.5\textwidth]{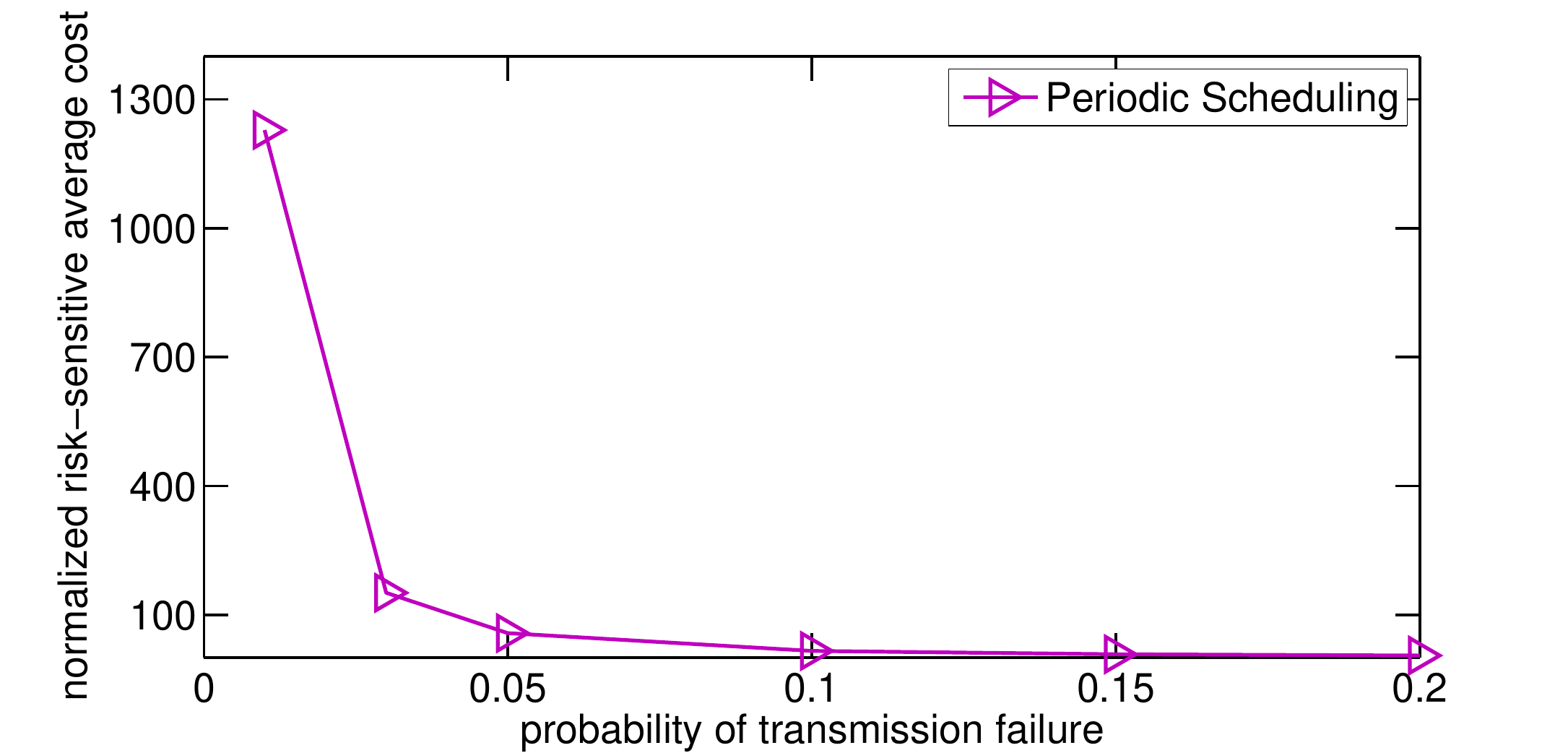}}
	\caption{In a multi-client scenario, the normalized risk-sensitive average cost (normalized by the cost of the optimal policy) vs. the failure transmission parameter $\epsilon$ is shown. (The parameters are $N=3$, $p_1=p_2=p_3=1-\epsilon$, $\tau_1=4$, $\tau_2=6$, $\tau_3=8$, $\theta=0.05$.)}
	\label{fg:case_3}
\end{figure}

\section{Conclusions}
In this paper we have addressed the issue of designing scheduling policies in order to support inter-delivery requirements of wireless clients in cyber-physical systems.
A novel risk-sensitive approach has been employed to penalize the ``exceedance" over the allowable thresholds of inter-delivery times.

The resulting MDP that involves infinitely many states can be reduced to an equivalent MDP which involves only a finite number of states, thus showing that
a stationary optimal  policy exists when the risk-sensitive parameter $\theta$ is sufficient small. Based on this, we have designed a finite time algorithm to obtain the optimal policy. 

To address the curse of dimensionality from MDP approach, we proposed a high-reliability asymptotic approach, and derived optimal policy for two-client scenario in the high-reliability regime. Further, we have designed an SN policy for a general number of clients based on our analysis result.
The simulation results show that the proposed policies provide near-optimal performance even for moderately large values of the failure probabilities, justifying the approach.

\section*{Acknowledgments}
This work is sponsored
in part by the National Basic Research Program of China
(2012CB316001), the Nature Science Foundation of China
(61201191, 61401250, 61321061, 61461136004), NSF under Contract Nos.
CNS-1302182 and CCF-0939370, AFOSR under Contract No.
FA-9550-13-1-0008, and Hitachi Ltd.

\bibliographystyle{abbrv}
\bibliography{scheduleproblem}  

\appendix

\section{Non-exclusionary Policies are Not Restrictive}\label{se:policies of interest}
We recall that Non-Exclusionary (NE) policies are those stationary policies which do not serve a client $n$ when the system state $\mathbf{x}$ is $(\tau_1,\cdots,\tau_{n-1},0,\tau_{n+1},\cdots,\tau_N)$.
We next show that for a non-NE, either there is an NE policy out-performing it, or it is trivial to derive its cost function.

For a client $n$, denote  $\mathbf{x}^{(na)}$ as the system state \\ $(\tau_1,\cdots,\tau_{n-1},a,\tau_{n+1},\cdots,\tau_N)$ for integer $a\in[0,\tau_n]$. Similarly, for $n\neq l$, denote $\mathbf{x}^{(na_1, la_2)}$ as the state $\mathbf{x}$ such that $x_n=a_1$, $x_l=a_2$, and  $x_j=\tau_j,\forall j\neq n,l$.
In a $T$-horizon MDP-$2$ problem, denote by $\pi$ the policy which transmits client $n$ in the first time slot, and then follows the optimal policy. Similarly, $\tilde{\pi}$ as the policy which transmits client $l$ in the first time slot and then follows the optimal policy. Then $V^{\pi}_T(\mathbf{x}),V^{\tilde{\pi}}_T(\mathbf{x})$ are the costs associated with these two policies with any initial state $\mathbf{x}$, respectively. 

\begin{app_lemma}\label{th:NE lemma1}
If there exists a client $l$ such that $p_l>p_n$, then for each time slot $T$, we have  $V^{\tilde{\pi}}_T(\mathbf{x}^{(n0)})\leq V^{\pi}_T(\mathbf{x}^{(n0)})$. That is, if serving client $n$ in state $\mathbf{x}^{(n0)}$ is optimal,  serving client $l$ is also optimal.
\end{app_lemma}

\begin{app_lemma}\label{th:NE lemma2}
When $p_n=\min_l{p_l}$,
if the optimal action in state $\mathbf{x}^{(n0)}$ with $T$ time slots to go is to serve client $n$, then the optimal action in state $\mathbf{x}^{(na)},\forall a\in\{1,\cdots,\tau_n \}$ with $T$ time slots to go is also to serve client $n$.
\end{app_lemma}

The proofs for Lemma \ref{th:NE lemma1} and Lemma \ref{th:NE lemma2} are omitted due to space constraints.

Combining Lemma \ref{th:NE lemma1} and Lemma \ref{th:NE lemma2}, for a non-NE stationary policy $f$ which serves client $n$ in state $\mathbf{x}^{(n0)}$, either there exists an NE policy which out-performs $f$ (as in Lemma \ref{th:NE lemma1}), or $f$ keeps serving client $n$ (after hitting the state $\boldsymbol{\tau}$) following Lemma \ref{th:NE lemma2} and therefore has a trivially computable cost.

\end{document}